\documentclass[reqno]{amsart}

\newcommand{\K}{\mathbb{K}}

\newcommand{\PP}{\mathbb{P}}

\newcommand{\Hilb}[2]{\textnormal{\bf Hilb}^{#1}_{#2}}

\newcommand{\Proj}{\textnormal{Proj}\,}

\newcommand{\sat}{\textnormal{sat}}
\newcommand{\reg}{\textnormal{reg}}

\newcommand{\up}[1]{\textnormal{\textsf{e}}^{+}_{#1}}
\newcommand{\down}[1]{\textnormal{\textsf{e}}^{-}_{#1}}
\newcommand{\pos}[2]{\mathcal{P}({#1},{#2})}

\newcommand{\restrict}[2]{#1_{(\geqslant{#2})}}

\newcommand{\DegLex}{\textnormal{\texttt{DegLex}}}

\usepackage[english]{babel}

\usepackage{mathrsfs}
\usepackage{amsmath,amssymb,enumerate,array}
\usepackage{tikz}
\usepackage[colorlinks]{hyperref}

\usepackage{algorithmic}
\usepackage[boxed]{algorithm}

\newtheorem{theorem}{Theorem}[section]

\newtheorem{proposition}[theorem]{Proposition}
\newtheorem{lemma}[theorem]{Lemma}

\newtheorem{definition}[theorem]{Definition}

\newtheorem{example}{Example}[section]

\begin{document}
\title[An efficient implementation of the algorithm computing the Borel-fixed points of $\Hilb{n}{p(t)}$]{An efficient implementation of the algorithm computing the Borel-fixed points of a Hilbert scheme}

\author[P.~Lella]{Paolo Lella}
\address{Dipartimento di Matematica dell'Universit\`{a} di Torino\\
         Via Carlo Alberto 10,
         10123 Torino, Italy}
\email{\href{mailto:paolo.lella@unito.it}{paolo.lella@unito.it}}
\urladdr{\href{http://www.personalweb.unito.it/paolo.lella/}{www.personalweb.unito.it/paolo.lella/}}


\begin{abstract}
Borel-fixed ideals play a key role in the study of Hilbert schemes. Indeed each component and each intersection of components of a Hilbert scheme contains at least one Borel-fixed point, i.e. a point corresponding to a subscheme defined by a Borel-fixed ideal. Moreover Borel-fixed ideals have good combinatorial properties, which make them very interesting in an algorithmic perspective. In this paper, we propose an implementation of the algorithm computing all the saturated Borel-fixed ideals with number of variables and Hilbert polynomial assigned, introduced from a theoretical point of view in the paper \lq\lq Segment ideals and Hilbert schemes of points\rq\rq, Discrete Mathematics 311 (2011).
\end{abstract}

\maketitle

\section{Introduction}
The Hilbert scheme $\Hilb{n}{p(t)}$ parametrizes all the subschemes and all the families of subschemes of the projective space $\PP^n$ with Hilbert polynomial $p(t)$. Borel-fixed ideals are a basic tool for the direct study of Hilbert schemes, because 
\begin{itemize}
\item each component and each intersection of components of a Hilbert scheme contains at least one point defined by a Borel-fixed ideals;
\item the Borel-fixed ideals have a strong combinatorial property which makes them very convenient also from an algorithmic perspective.
\end{itemize}
 For instance, Hartshorne in his thesis \cite{HartshorneThesis} proved the connectedness of Hilbert scheme constructing sequences of deformations of Borel-fixed ideals (he called them \emph{balanced} ideals) which lead from any point of the Hilbert scheme to the point determined by the unique saturated lexicographic associated to the Hilbert polynomial $p(t)$.

More recently, many papers \cite{BCLR,BLR,BLMR,CLMR,CioffiRoggero,LellaDeformations,LellaRoggero} by Bertone, Cioffi, Marinari, Roggero and the author showed how to use Borel-fixed ideals for a local study of the Hilbert scheme, mainly constructing families of ideals sharing with a fixed Borel-fixed ideals the same basis of the quotient space.

Therefore it was very important to have an algorithm computing for each Hilbert scheme $\Hilb{n}{p(t)}$ all the points defined by Borel-fixed ideals.  An algorithm for computing all the Borel-fixed ideals in $\K[x_0,\ldots,x_n]$ with Hilbert polynomial $p(t)$ based on a combinatorial approach has been proposed in \cite{CLMR} and another algorithm with 
slight differences to this, and published later, is presented in \cite{MooreNagel}.

In this paper, we describe a concrete implementation of the algorithm, that turns out to be very efficient as we will show with an experimental analysis.

\section{Notation and general setting}

We will consider a field $\K$ of characteristic 0 and for any polynomial ring $\K[x] := \K[x_0,\ldots,x_n]$ we will order the variables as $x_n > \ldots > x_0$. Following the notation of \cite{Gotzmann}, we will refer to the Hilbert polynomial of a homogeneous ideal $I \subset \K[x_0,\ldots,x_n]$ as the Hilbert polynomial of the graded module $\K[x_0,\ldots,x_n]/I$, i.e.
\[
p(t) = \dim_\K \K[x_0,\ldots,x_n]_t/I_t,\qquad t \gg 0.
\]

A homogeneous ideal $I$ is said Borel-fixed if it is fixed by the action of the Borel subgroup of upper triangular matrices. Looking at the action of the elements of the Borel subgroup $\mathrm{Id}+E_{i,j}$, where $i<j$ and $E_{i,j}$ is a matrix with all entries equal to 0 except for the entry of the $i$-th row and $j$-h column equal to 1, it is possible to prove that the ideal $I$ has a nice combinatorial property: in fact an ideal $I$ is Borel-fixed if and only if it is a monomial ideal and
\[
x^\alpha \in I\quad \Longrightarrow\quad \frac{x_j}{x_i} x^\alpha \in I,\ \forall\ x_i \mid x^\alpha,\ x_j > x_i.
\]
As done in \cite{CLMR,LellaDeformations}, we define the elementary Borel moves: 
\begin{itemize}
\item $\down{j}$ as the element $\frac{x_{j-1}}{x_j}$ in the field of fraction $\K(x)$ of $\K[x]$;
\item $\up{i}$ as the element $\frac{x_{i+1}}{x_i} \in \K(x)$;
\end{itemize}
and for any monomial $x^\alpha$ we will say that $\down{i}$ (resp. $\up{j}$) is admissible on $x^\alpha$ if $\down{i}(x^\alpha)=\frac{x_{i-1}}{x_i} x^\alpha \in \K[x]$ (resp. $\up{j}(x^\alpha)=\frac{x_{j+1}}{x_j} x^\alpha \in \K[x]$). We will use additive notation to denote the composition of an elementary move with itself, for instance $2\down{j} = \down{j}\circ \down{j} = \left(\frac{x_{j-1}}{x_j}\right)^2$.

$\pos{n}{r}$ will denote the poset of all monomials in $\K[x]_r$ with the Borel partial order $\leq_B$ given by the transitive closure of the order relations
\[
\up{i}(x^\alpha) >_B x^\alpha >_B \down{j}(x^\alpha)
\]
$\forall\ x^\alpha \in \K[x]_r,\ \forall\ \up{i},\down{j}$ admissible, and we call Borel set any subset of a poset $\pos{n}{r}$ closed under increasing elementary moves. By definition, given a Borel-fixed ideal $I \subset \K[x]$, the monomial basis of each homogeneous piece $I_t$ of fixed degree will define a Borel set, thus we will write $\{I_t\}$ meaning the Borel set of $\pos{n}{t}$ defined by the monomials in $I_t$. On the other hand, given a Borel set $\mathscr{B} \subset \pos{n}{r}$, we will denote by $\langle\mathscr{B}\rangle$ the Borel-fixed ideal generated by the monomial of $\mathscr{B}$ and by $\langle\mathscr{B}\rangle^\sat$ its saturation.

For any monomial $x^\alpha$ we denote 
\begin{itemize}
\item $\max x^\alpha = \max \{x_i \text{ s.t. } x_i \mid x^\alpha\}$;
\item $\min x^\alpha = \min \{x_j \text{ s.t. } x_j \mid x^\alpha\}$.
\end{itemize}
For the order chosen on the variables, it will not be misleading to write only the index of the variable instead of the variable itself.

For any Borel set $\mathscr{B} \subset \pos{n}{r}$ we will denote by $\mathscr{B}^\mathcal{C}$ the complement set $\pos{n}{r}\setminus\mathscr{B}$. Obviously $\mathscr{B}^{\mathcal{C}}$ is closed under decreasing elementary moves and we will call such a set \emph{order set}, being its dehomogeneization (imposing $x_0=1$) an order ideal.

For any subset $S$ of $\pos{n}{r}$, $\restrict{S}{i}$ will denote the subset of $S$ of the monomials with minimum greater than or equal to $i$:
\[
\restrict{S}{i} := \left\{ x^\alpha \in S\ \vert\ \min	x^\alpha \geqslant i\right\}.
\]
Obviously $\restrict{S}{0} = S$.

As well known, given any Borel-fixed ideal $I$, its saturation $I^\sat$ is generated by the monomials obtained substituting $x_0=1$ in the monomials that generate $I$. Moreover, $x_0$ is not a nihilpotent element for any Borel-fixed ideals.

\begin{proposition}\label{prop:hyperplaneSection}
Let $I \subset \K[x]$ be a Borel-fixed ideal. The linear form $x_0$ is regular for $I$. Thus, for any shift of degree $t$, there is the short exact sequence induced by the multiplication by $x_0$
\begin{equation}\label{eq:hyperplaneSection}
0\longrightarrow\ \dfrac{\K[x]}{I}(t-1) \stackrel{\cdot x_0}{\longrightarrow} \dfrac{\K[x]}{I}(t) \longrightarrow\dfrac{\K[x]}{(I,x_0)}(t) \longrightarrow 0.
\end{equation}
\end{proposition}

This regular sequence says that the ideal $(I,x_0)$ has Hilbert polynomial $\Delta p(t) = p(t) - p(t-1)$. We will perform hyperplane section repeatedly so we define $\Delta^0 p(t) := p(t)$ and recursively $\Delta^k p(t) := \Delta^{k-1}p(t) - \Delta^{k-1}p(t-1)$. It is very easy to manipulate Hilbert polynomial considering the Gotzmann decomposition.
\begin{definition}\index{Gotzmann number}
An admissible Hilbert polynomial $p(t)$ has a unique \emph{Gotzmann decomposition}
\begin{equation}\label{eq:GotzmannDecomposition}
p(t) = \binom{t+a_1}{a_1}  + \ldots + \binom{t+a_r - (r-1)}{a_r},
\end{equation}
$a_1 \geqslant \ldots \geqslant a_r$. The number $r$ of terms in this sum is said \emph{Gotzmann number} of $p(t)$.
\end{definition}

It can be easily proved that the Gotzmann decomposition of $\Delta p(t)$ can be obtained from the decomposition of $p(t)$ decreasing by 1 each index $a_i$ and discarding the binomial coefficients with a negative index below. Moreover we define an inverse operator $\Sigma$ that associates to $p(t)$ the Hilbert polynomial obtained increasing by 1 all the indices $a_i$.

\begin{example}
The Hilbert polynomial $p(t) = 3t+1$ has Gotzmann decomposition
\small
\[
\binom{t+1}{1} + \binom{t+1-1}{1} + \binom{t+1-2}{1} + \binom{t-3}{0}.
\]
\normalsize
Hence
\small
\[
\begin{split}
\Delta p(t) &{}= \binom{t}{0} + \binom{t-1}{0} + \binom{t-2}{0} = 3\\
\Sigma p(t) &{} = \binom{t+2}{2} + \binom{t+2-1}{2} + \binom{t+2-2}{2} + \binom{t+1-3}{1} = \\
&{} = \frac{3}{2}t^2 + \frac{5}{2}t -1.
\end{split}
\]
\normalsize
\end{example}

We conclude the background materials with the properties linking the Hilbert polynomial to the regularity of an ideal.
\begin{definition}
A coherent sheaf $\mathcal{F}$ over $\PP^n$ is said $m$-regular if for every $i > 0$
\begin{equation}
H^i\big(\mathcal{F}(m-i)\big) = 0.
\end{equation}
The Castelnuovo-Mumford regularity of $\mathcal{F}$ is the smallest integer $m$, for which $\mathcal{F}$ is $m$-regular.
\end{definition}

\begin{theorem}[\textbf{Gotzmann's Regularity Theorem}]\label{th:RegularityTheorem}Let $A$ be any $\K$-algebra and let $Z \subset \Proj A[x_0,\ldots,x_n]$ be any subscheme with Hilbert polynomial $p(t)$, whose Gotzmann number is $r$. Then the sheaf of ideals $\mathcal{I}_{Z}$ is $r$-regular.
\end{theorem}

\begin{proposition}[{\cite[Proposition 2.6]{GreenGIN}}]\label{prop:regularityVsCMregularity}
Let $I$ be a saturated ideal in $\K[x]$. The regularity of $I$ is equal to the Castelnuovo-Mumford regularity of the sheaf of ideals $\mathcal{I}$ obtained from the sheafification of $I$.
\end{proposition}

\begin{proposition}[{\cite{GreenGIN,BayerStillmanREG}}]\label{prop:regularityDegreeGenerators}
The regularity of a Borel-fixed ideal $I \subset \K[x_0,\ldots,x_n]$ is equal to the maximal degree of one of its generators.
\end{proposition}

\section{Main properties}

Now we recall the technical properties that we need to guarantee the correctness of the algorithm, which are already proved in \cite{CLMR} and that we here adapt at the current notation.

\begin{lemma}\label{lem:removingMinimal}
Let $J\subset \K[x]$ be a saturated Borel-fixed with Hilbert polynomial $p(t)$ whose Gotzmann number is $r$. 
Let $x^{\beta}$ be a minimal monomial for $\leq_B$ of $\{J_r\} \subset \pos{n}{r}$ such that $\min x^{\beta} = x_0$. Then the ideal  $I =\langle \{J_r\}\setminus \{x^{\beta}\}\rangle^\sat$ is Borel-fixed with Hilbert polynomial $\overline{p}(t) = p(t)+1$.
\end{lemma}
\begin{proof}
First, note that by definition of minimal monomial, $\{I_r\}$ is still a Borel set. Called $q(t)$ the volume polynomial of $J$, we show that $I$ has volume polynomial $\overline{q}(t) = q(t)-1$ applying Gotzmann's Persistence Theorem \cite[Theorem 3.8]{GreenGIN}, i.e. proving that 
\[
\dim_{\K} J_r - \dim_{\K} I_r = \dim_{\K} J_{r+1} - \dim_{\K} I_{r+1} = 1.
\]  

By construction $\dim_{\K} J_r - \dim_{\K} I_r = 1$. The Borel condition ensures that $x^\beta x_0 \in J_{r+1} \setminus I_{r+1}$ and there are no other elements, because $x^\beta x_0 $ is the only monomial that cannot be generated from the monomials in $I_{r}$ by multiplication of a single variable.  In fact let us consider the monomial $x_i x^\beta$, $i > 0$. Since $x_0 \mid x^\beta$ the following identity holds:
\[
x_i x^\beta = \dfrac{x_i}{x_{i-1}} \cdot \ldots \cdot \dfrac{x_1}{x_0} x_0 x^\beta = \up{i-1}\circ \cdots \circ \up{0}(x^\beta) x_0
\]
and for each $i$, $\up{i-1}\circ \cdots \circ \up{0}(x^\beta)$ belongs to $I_r$, by the minimality of $x^\beta$.
\end{proof}

\begin{proposition}[{\cite[Proposition 5.2]{CLMR}}]\label{prop:difference}
Consider a saturated Borel-fixed ideal $J \subset \K[x]$ defining a subscheme with Hilbert polynomial $p(t)$ whose Gotzmann number is $r$. Let $I=J\vert_{x_1=x_0=1}$ be its double saturation and let $\overline{p}(t)$ be the Hilbert polynomial of the subscheme defined by $I$ in $\PP^n$. Then
\begin{equation}
\overline{p}(t)=p(t)- \dim_\K I_r + \dim_\K J_r.
\end{equation}
\end{proposition}

\begin{theorem}\label{cor:bijectionAllIdeals}
Let $p(t)$ be an admissible Hilbert polynomial in $\mathbb{P}^n$ and let $r$ be an integer greater than or equal to the Gotzmann number of $p(t)$. There is a bijective function between the set $\mathcal{I}^{n}_{p(t)}$ of the saturated Borel-fixed ideals of $\K[x_0,\ldots,x_n]$ with Hilbert polynomial $p(t)$ and the set
\[ 
\mathcal{B}^n_{p(t)} = \left\{\begin{array}{c} \mathscr{B} \subset \pos{n}{r} \text{ Borel set s.t.} \\ \left\vert \restrict{\mathscr{B}^{\mathcal{C}}}{i}\right\vert =  \Delta^i p(r),\ \forall\ i = 0,\ldots,n-1\end{array}\right\}.
\]
More precisely
\[
\begin{array}{ccc}
\mathcal{I}^n_{p(t)} & \stackrel{1:1}{\longleftrightarrow} & \mathcal{B}^n_{p(t)}\\
J & \longrightarrow & \{J_r\}\\
\langle \mathscr{B} \rangle^{\sat} & \longleftarrow & \mathscr{B}.
\end{array}
\]
\end{theorem}
\begin{proof}
First of all note that if the two maps are well-defined, i.e. for each $J$, $J_{\geqslant r} = \langle J_r \rangle$ (Proposition \ref{prop:regularityDegreeGenerators} and Gotzmann's Regularity Theorem), 
\begin{eqnarray*}
& J\ \longrightarrow\ \{J_r\}\ \longrightarrow\ \big\langle\{J_r\}\big\rangle^{\sat} = J, &\\
& \mathscr{B}\ \longrightarrow\ \langle \mathscr{B} \rangle^{\sat}\ \longrightarrow\  \big\{\langle \mathscr{B} \rangle^{\sat}_r\big\} = \mathscr{B}.
\end{eqnarray*}

Let $J\subset \K[x]$ be a Borel-fixed with Hilbert polynomial $p(t)$ and let $\mathscr{N}= \pos{n}{r}\setminus \{J_r\}$. Obviously $\big\vert\restrict{\mathscr{N}}{0}\big\vert = \vert\mathscr{N}\vert = p(r)= \Delta^0 p(r)$. Using the short exact sequence \eqref{eq:hyperplaneSection}, we determine the Borel ideal $I = (J,x_0)^{\sat} \cap \K[x_1,\ldots,x_n]$ with Hilbert polynomial $\Delta p(t)$.
Thus being $\{I_r\} = \restrict{\{J_r\}}{1} \subset \pos{n-1}{r}$, $\left\vert \restrict{\mathscr{N}}{1} \right\vert = \left\vert \{I_r\}^{\mathcal{C}} \right\vert = \Delta p(r)$. Since $I$ is Borel-fixed in the polynomial ring $\K[x_1,\ldots,x_n]$ we can repeat the reasoning with the hyperplane section defined by $x_1 = 0$ and so on.

Let us now consider a Borel set $\mathscr{B} \subset \pos{n}{r}$, such that the complement $\mathscr{N} = \mathscr{B}^{\mathcal{C}}$ satisfies the condition $\left\vert\restrict{\mathscr{N}}{i}\right\vert = \Delta^{i} p(r)$ for every $i$. Firstly $\reg\left(\langle\mathscr{B}\rangle^{\sat}\right) \leqslant r$ by Proposition \ref{prop:regularityDegreeGenerators}, so let us prove that $\langle\mathscr{B}\rangle^\sat$ has Hilbert polynomial $p(t)$. We proceed by induction on the degree $d$ of the Hilbert polynomial. For any $n$, if $\deg p(t) = 0$, then $\restrict{\mathscr{N}}{i} = \emptyset$, for every $i \geqslant 1$, since $\Delta p(t) = 0$, that is for any $x^\beta \in \mathscr{N}$, $\min x^\beta = 0$. Applying repeatedly Lemma \ref{lem:removingMinimal} starting from the Hilbert polynomial $\overline{p}(t) = 0$ (corresponding to the ideal $(1)$), we obtain that $\langle \mathscr{B}\rangle^{\sat}$ defines a module $\K[x]/\langle \mathscr{B}\rangle^{\sat}$ having constant Hilbert polynomial $p(t) (= p(r))$. 

Let us know suppose that the map $\mathscr{B} \rightarrow \langle\mathscr{B}\rangle^{\sat}$ is well-defined for any Hilbert polynomial of degree $d-1$ and let $p(t)$ be a Hilbert polynomial of degree $d$. $\overline{\mathscr{B}} = \restrict{\mathscr{B}}{1} \subset \pos{n-1}{r}$ realizes the condition of the theorem w.r.t. the Hilbert polynomial $\Delta p(t)$ and $\deg \Delta p(t) = d-1$. Hence by the inductive hypothesis the ideal $\langle \overline{\mathscr{B}}\rangle^\sat \subset \K[x_1,\ldots,x_n]$ has Hilbert polynomial $\Delta p(t)$. Let $\overline{p}(t)$ be the Hilbert polynomial of $\langle \overline{\mathscr{B}}\rangle^\sat$ in $\K[x_0,\ldots,x_n]$: $\overline{p}(t) = p(t) + a$, because $\Delta \overline{p}(t) = \Delta p(t)$. $\langle \overline{\mathscr{B}}\rangle^\sat$ turns out to be the $x_1$-saturation of $\langle \mathscr{B}\rangle^{\sat}$, so by Proposition \ref{prop:difference} the Hilbert polynomial of $\K[x_0,\ldots,x_n]/\langle \mathscr{B}\rangle^{\sat}$ differs by a constant from $\overline{p}(t)$ and since $\vert \mathscr{N} \vert = \vert \restrict{\mathscr{N}}{0}\vert = p(r)$ it coincides with $p(t)$.
\end{proof}

\section{The algorithm}
Therefore to compute the saturated Borel-fixed ideals we can construct Borel sets with the prescribed property. The proof of Theorem \ref{cor:bijectionAllIdeals} suggests to use a recursive algorithm: i.e. to determine the Borel sets in $\pos{n}{r}$ corresponding to the Hilbert polynomial $p(t)$, we begin computing the Borel sets in $\pos{n-1}{r}$ corresponding to the Hilbert polynomial $\Delta p(t)$. 

Let $\overline{\mathscr{B}} \subset \pos{n-1}{r}$ a Borel set corresponding to the Hilbert polynomial $\Delta p(t)$ and let $\overline{\mathscr{N}} = \overline{\mathscr{B}}^{\mathcal{C}}$. In order for $\overline{\mathscr{B}}$ to be the restriction $\restrict{\mathscr{B}}{1}$ of a Borel set $\mathscr{B} \subset \pos{n}{r}$ (where $\pos{n}{r}$ contains one more variable smaller than variables in $\pos{n-1}{r}$), each monomial that can be obtained by decreasing moves from a monomial in $\overline{\mathscr{N}}$ has to belong to $\mathscr{N} = \mathscr{B}^{\mathcal{C}}$. This extension of an order set $\overline{\mathscr{N}} \subset \pos{n-1}{r}$ to an order set $\mathscr{N} \subset \pos{n}{r}$ has an ideal interpretation.

\begin{lemma}[{\cite[Lemma 5.1]{CLMR}}]\label{lem:partition} Let $\overline{\mathscr{B}} \subset \pos{n-1}{r}$ be a Borel set and let $\overline{\mathscr{N}} = \overline{\mathscr{B}}^{\mathcal{C}}$. Moreover let $\mathscr{N} \subset \pos{n}{r}$ be the order set containing the monomials in $\overline{\mathscr{N}}$ and all those obtained by decreasing moves from them. Then,
\begin{equation}
\mathscr{N} = \pos{n}{r} \setminus \Big\{ \big(\langle\overline{\mathscr{B}}\rangle^{\sat} \cdot \K[x_0,\ldots,x_n]\big)_r \Big\}.
\end{equation}
\end{lemma} 
\begin{proof}
Let us call $\mathscr{B}$ the Borel set containing the monomials of degree $r$ belonging to the ideal $\langle\overline{\mathscr{B}}\rangle^{\sat} \cdot \K[x_0,\ldots,x_n]$.
Let $x^\alpha = x_n^{\alpha_n}\cdots x_0^{\alpha_0}$ be a monomial of $\pos{n}{r}$ and suppose $\min x^\alpha = 0$, i.e. $\alpha_0 > 0$. The monomial $\alpha_0\up{0}(x^\alpha) = x_n^{\alpha_n}\cdots x_1^{\alpha_1 + \alpha_0}$\hfill belongs\hfill to\hfill $\pos{n-1}{r}$,\hfill so\hfill either\hfill $\alpha_0\up{0}(x^\alpha) \in \overline{\mathscr{B}}$\hfill or $\alpha_0\up{0}(x^\alpha) \in \overline{\mathscr{N}}$.

If\hfill $x_n^{\alpha_n}\cdots x_1^{\alpha_1 + \alpha_0} \in \overline{\mathscr{B}}$,\hfill then\hfill $x_n^{\alpha_n}\cdots x_2^{\alpha_2}$\hfill is\hfill in\hfill $\langle \overline{\mathscr{B}} \rangle^{\sat}$\hfill and\hfill so\hfill $x^\alpha \in \langle\overline{\mathscr{B}}\rangle^{\sat} \cdot \K[x_0,\ldots,x_n]$,\hfill otherwise\\ $x_n^{\alpha_n}\cdots x_1^{\alpha_1 + \alpha_0} \in \overline{\mathscr{N}}$ implies $x^{\alpha} \in \mathscr{N}$.
\end{proof}

By Proposition \ref{prop:difference}, we know that the Hilbert polynomial corresponding to a Borel set $\mathscr{B}$ of the type $\{(\langle\overline{\mathscr{B}}\rangle^{\sat}\cdot \K[x])_r\}$ differs from the target Hilbert polynomial by a constant: to determine this constant we compare the value $p(r)$ of the Hilbert polynomial $p(t)$ in degree $r$ with the cardinality of the order set $\mathscr{N}$ obtained by decreasing moves from $\overline{\mathscr{N}}$.

\begin{lemma}\label{lem:countingConditions}
Let $\overline{\mathscr{N}} \subset \pos{n-1}{r}$ be an order set and let $\mathscr{N} \subset \pos{n}{r}$ be the order set defined from $\overline{\mathscr{N}}$ by decreasing moves. Then,
\begin{equation}\label{eq:cardN}
\left\vert \mathscr{N} \right\vert = \sum_{\begin{subarray}{c} x^{\alpha} \in \overline{\mathscr{N}} \\ x^{\alpha}=x_n^{\alpha_n}\cdots x_{1}^{\alpha_1}\end{subarray}} (\alpha_1 + 1).
\end{equation}
\end{lemma}
\begin{proof}
Each monomial $x^{\alpha} \in \overline{\mathscr{N}}$ imposes the belonging to $\mathscr{N}$ (in addition to itself) of the monomials obtained from it applying the decreasing moves $\{\down{1},2\down{1},\ldots,\alpha_1\down{1}\}$.
\end{proof}

There are three possibilities:
\begin{itemize}
\item $p(r) - \vert \mathscr{N} \vert < 0$, $\overline{\mathscr{N}}$ imposes too many monomials ouside the ideal, so the hyperplane section defined by $\langle \overline{\mathscr{B}} \rangle^{\sat}$ has to be discarded (there exist no Borel-fixed ideals corresponding to $p(t)$ with such a hyperplane section);
\item $p(r) - \vert \mathscr{N} \vert = 0$, $\langle\overline{\mathscr{B}} \rangle^{\sat}  \subset \K[x_0,\ldots,x_n]$ is one of the ideals sought;
\item $p(r) - \vert \mathscr{N} \vert > 0$, applying repeatedly Lemma \ref{lem:removingMinimal} we determine the ideals we are looking for.
\end{itemize}

\begin{proposition}[Cf. {\cite[Remark 5.3]{CLMR}}]\label{prop:estimateRemovals}
Let $p(t)$ be an admissible Hilbert polynomial with Gotzmann number $r$ and let $r_1$ be the Gotzmann number of $\Delta p(t)$.
\begin{enumerate}[(i)]
\item\label{it:estimateRemovals_i} Given a saturated Borel-fixed ideal $J \subset \K[x_1,\ldots,x_n]$ such that $\K[x_1,\ldots,x_n]/J$ has Hilbert polynomial $\Delta p(t)$, to pass from $\{(J\cdot\K[x_0,\ldots,x_n])_r\}\subset \pos{n}{r}$ to a Borel set corresponding to $p(t)$, we need to remove at most $r-r_1$ monomials.
\item\label{it:estimateRemovals_ii} We need to remove exactly $r-r_1$ monomials if we consider the lexicographic ideal $L \subset \K[x_1,\ldots,x_n]$ corresponding to the polynomial $\Delta p(t)$.
\end{enumerate}
\end{proposition}
\begin{proof}\emph{(\ref{it:estimateRemovals_i})}
The minimal Hilbert polynomial having first difference equal to $\Delta p(t)$ is $\Sigma(\Delta p)(t)$ and the Gotzmann number of $\Sigma(\Delta p)(t)$ coincides with the Gotzmann number of $\Delta p(t)$.

\emph{(\ref{it:estimateRemovals_ii})} The Hilbert polynomial $\Sigma(\Delta p)(t)$ is admissible so there exists the saturated lexicographic ideal $L$ with such Hilbert polynomial. By construction $L$ is the $1$-lifting of the ideal defining its hyperplane section, therefore it has no generators involving the variable $x_1$ and the ideal $L \cap \K[x_1,\ldots,x_n]$, being generated by the same monomials, is still a lexicographic ideal.
\end{proof}

\subsection{The pseudocode description}
The recursive strategy naturally gives rise to a rooted tree where the nodes are all the Borel-fixed ideals met during the computation and the father of each node is the ideal of its hyperplane section (see Figure \ref{fig:tree} for an example). 
This graph turns out to be a rooted tree because with a sequence of hyperplane section we obtain a unique ideal from any Borel-fixed ideal with the given Hilbert polynomial:
\begin{itemize}
\item the ideal $(1) \subset \K[x_{d+1},\ldots,x_n]$ applying $d+1$ sections, whenever the degree of the Hilbert polynomial $d$ is smaller than $n-1$;
\item the ideal $(x_n^{c}) \subset \K[x_{n-1},x_n]$, where $c = \Delta^{d} p (r)$, if $d = n-1$.
\end{itemize}
The leaves are the Borel-fixed ideal we are looking for and by definition this tree will have height $d+2$ if $d < n-1$ and $d+1$ if $d=n-1$.

To compute the Borel-fixed ideals in $\K[x_0,\ldots,x_n]$ with Hilbert polynomial $p(t)$, the algorithm described in \cite{CLMR} requires to compute the ideals in $\K[x_1,\ldots,x_n]$ with Hilbert polynomial $\Delta p(t)$ and so on, until the last recursive call that requires to compute the ideal representing the root of the tree. Now, starting from the root of the tree, the algorithm performs a Breadth-First-Search visit. This approach is not optimal, because it requires much space in the memory to store also all the ideal which are not leaves at maximal height and so which are necessary only temporarily. Moreover some of such ideals can be discarded immediately because too many conditions are imposed when embedded in a polynomial ring with one more variable.
For this reason we prefer an algorithm that visits the tree with a Depth-First-Search approach. The function \textsc{BorelFixedIdealsGenerator} (\textbf{Algorithm} \ref{alg:BorelGeneratorDFS}) initializes the computation determining the root of the tree and starting the DFS visit.

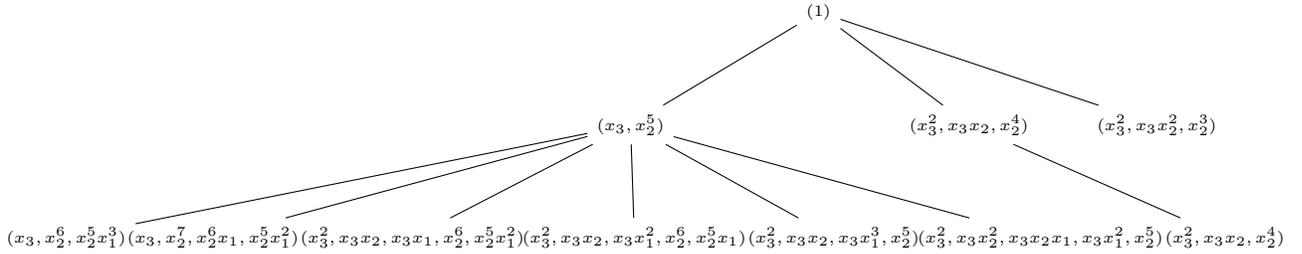
\begin{figure*}[!ht]
\begin{center}
\begin{tikzpicture}[>=latex]
\node (c11) at (0,0) [] {\tiny $(x_3,x_2^6,x_2^5x_1^3)$};
\node (c12) at (1.95,0) [] {\tiny $(x_3,x_2^7,x_2^6x_1,x_2^5x_1^2)$};
\node (c13) at (4.6,0) [] {\tiny $(x_3^2,x_3x_2,x_3x_1,x_2^6,x_2^5x_1^2)$};
\node (c14) at (7.55,0) [] {\tiny $(x_3^2,x_3x_2,x_3x_1^2,x_2^6,x_2^5x_1)$};
\node (c15) at (10.2,0) [] {\tiny $(x_3^2,x_3x_2,x_3x_1^3,x_2^5)$};
\node (c16) at (12.95,0) [] {\tiny $(x_3^2,x_3x_2^2,x_3x_2x_1,x_3x_1^2,x_2^5)$};
\node (c21) at (15.4,0) [] {\tiny $(x_3^2,x_3x_2,x_2^4)$};

\node (b1) at (7.5,1.5) [] {\tiny $(x_3,x_2^5)$};
\node (b2) at (12,1.5) [] {\tiny $(x_3^2,x_3x_2,x_2^4)$};
\node (b3) at (14.5,1.5) [] {\tiny $(x_3^2,x_3x_2^2,x_2^3)$};

\node (a) at (10,3) [] {\tiny $(1)$};

\draw [-,thin] (a) -- (b1);
\draw [-,thin] (a) -- (b2);
\draw [-,thin] (a) -- (b3);

\draw [-,thin] (b1) -- (c11);
\draw [-,thin] (b1) -- (c12);
\draw [-,thin] (b1) -- (c13);
\draw [-,thin] (b1) -- (c14);
\draw [-,thin] (b1) -- (c15);
\draw [-,thin] (b1) -- (c16);

\draw [-,thin] (b2) -- (c21);

\end{tikzpicture}
\caption{\label{fig:tree}The tree associated at the computation of the Borel-fixed ideals in $\K[x_0,x_1,x_2,x_3]$ with Hilbert polynomial $p(t) = 5t-2$.}
\end{center}
\end{figure*}

\begin{algorithm}[!ht]
\caption{The procedure that initializes the computation detecting the root of the tree.}
\label{alg:BorelGeneratorDFS}
\begin{algorithmic}[1]
\STATE $\textsc{BorelFixedIdealsGenerator}\big(\K[x_0,\ldots,x_n],p(t)\big)$
\REQUIRE $\K[x_0,\ldots,x_n]$, polynomial ring.
\REQUIRE $p(t)$, admissible Hilbert polynomial in $\PP^{n}$.
\ENSURE the set of all Borel-fixed ideals in $\K[x_0,\ldots,x_n]$ with Hilbert polynomial $p(t)$.
\STATE $d \leftarrow \deg p(t)$;
\STATE $r \leftarrow \textsc{GotzmannNumber}\big(p(t)\big)$;
\IF{$d = n-1$}
\STATE $c \leftarrow \Delta^d p(t)$;
\STATE $\mathscr{B} \leftarrow \big\{(x_{n}^{c})_r\big\} \subset \pos{2}{r}$;\hfill \COMMENT{$\K[x_{n-1},x_n]$.}
\RETURN $\textsc{BorelIdeals}\big(\K[x_0,\ldots,x_n],p(t),d,\mathscr{B}\big)$;
\ELSE
\STATE $\mathscr{B} \leftarrow \pos{n-d}{r}$;\hfill \COMMENT{$\K[x_{d+1},\ldots,x_n]$.}
\RETURN $\textsc{BorelIdeals}\big(\K[x_0,\ldots,x_n],p(t),d+1,\mathscr{B}\big)$;
\ENDIF
\end{algorithmic}
\end{algorithm}

The core of the algorithm is described by the function \textsc{BorelIdeals} (\textbf{Algorithm} \ref{alg:coreDFS}). It takes a Borel set $\overline{\mathscr{B}}$ corresponding to a Borel-fixed ideal in $\K[x_k,\ldots,x_n]$ with Hilbert polynomial $\Delta^k p(t)$. If $k=0$ then $\overline{\mathscr{B}}$ corresponds to one of the searched ideal, so the function returns the saturation of the ideal generated by the monomials in $\overline{\mathscr{B}}$. Otherwise, the function embeds $\overline{\mathscr{B}}$ in a poset with one more variable obtaining $\mathscr{B}$, it computes how many monomials do not belong to $\mathscr{B}$ and how many monomials ($q$) have to be removed. Then it calls the function \textsc{Remove} to compute all the Borel set that can be obtained from $\mathscr{B}$ removing $q$ monomials and finally for each of such new Borel set we have the recursive call with index $k-1$.

\begin{algorithm}[!ht]
\caption{The core of the recursive strategy to compute Borel-fixed ideals in $\K[x]$ with Hilbert polynomial $p(t)$.}
\label{alg:coreDFS}
\begin{algorithmic}[1]
\STATE $\textsc{BorelIdeals}\big(\K[x_0,\ldots,x_n],p(t),k,\overline{\mathscr{B}}\big)$
\REQUIRE $\K[x_0,\ldots,x_n]$, polynomial ring.
\REQUIRE $p(t)$, admissible Hilbert polynomial in $\PP^{n}$.
\REQUIRE $k$, integer s.t. $0\leqslant k \leqslant \deg p(t)$.
\REQUIRE $\overline{\mathscr{B}}$, Borel set in $\pos{n-k}{r}$ s.t. the ideal $\langle\overline{\mathscr{B}}\rangle^\sat$ in $\K[x_{k},\ldots,x_n]$ has Hilbert polynomial $\Delta^{k} p(t)$.
\ENSURE the set of all saturated Borel-fixed ideals $J$ in $\K[x_0,\ldots,x_n]$ with Hilbert polynomial $p(t)$ s.t. $\restrict{\{J_r\}}{k} = \overline{\mathscr{B}}$.
\IF{$k = 0$}
\RETURN \label{line:saturation} $\langle\overline{\mathscr{B}}\rangle^\sat$;
\ELSE
\STATE\label{line:embed1} $\mathscr{B} \leftarrow \big\{\big(\langle\overline{\mathscr{B}}\rangle^\sat\cdot\K[x_{k-1},\ldots,x_n]\big)_r\big\} \subset \pos{n-k+1}{r}$;
\STATE\label{line:embed2} $q \leftarrow \Delta^{k-1} p(r) - \left\vert\mathscr{B}^{\mathcal{C}}\right\vert$;
\IF{$q\geqslant 0$}
\STATE $\textsf{oneMoreVariable} \leftarrow \textsc{Remove}(\mathscr{B},q,1)$;
\STATE $\textsf{ideals} \leftarrow \emptyset$;
\FORALL{$\widetilde{\mathscr{B}} \in \textsf{oneMoreVariable}$}
\STATE \parbox[t]{10cm}{$\textsf{ideals} \leftarrow \textsf{ideals} \cup {}$\\ \phantom{m} $\textsc{BorelIdeals}\big(\K[x_0,\ldots,x_n],p(t),k-1,\widetilde{\mathscr{B}}\big)$;}
\ENDFOR
\RETURN $\textsf{ideals}$;
\ELSE
\RETURN $\emptyset$;
\ENDIF
\ENDIF
\end{algorithmic}
\end{algorithm}

Also the function \textsc{Remove} (\textbf{Algorithm} \ref{alg:removeWithoutRepetions}) uses a recursive strategy, i.e. given a Borel set $\mathscr{B}$ and $q$ monomials to remove, for each minimal element $x^\alpha$ of $\mathscr{B}$, it calls itself on the Borel set $\mathscr{B}\setminus \{x^\alpha\}$ to which we have to remove $q-1$ monomials. Whenever we consider a Borel set $\mathscr{B}$ with (at least) two minimal elements $x^\alpha$ and $x^\gamma$ and we need to remove (at least) two monomials, the strategy just described would generate two times the same Borel set because $\big(\mathscr{B}\setminus\{x^\alpha\}\big)\setminus\{x^\gamma\} = \big(\mathscr{B}\setminus\{x^\gamma\}\big)\setminus\{x^\alpha\}$.
To avoid this repetition we add, as argument of the function \textsc{Remove}, a monomial $x^\beta$ and we require that any monomial removed from $\mathscr{B}$ is greater than $x^\beta$ w.r.t. a fixed term ordering (in \textbf{Algorithm} \ref{alg:removeWithoutRepetions} we chose $\DegLex$). So whenever the function is called by \textsc{BorelIdeals} we pass as argument the monomial 1 (any removal is admissible), whereas when \textsc{Remove} is called by itself, that is some removal has been already performed, we pass as argument the last monomial removed. In this way, assuming $x^\alpha >_{\DegLex} x^\gamma$, the function \textsc{Remove} will only generate the Borel set $\big(\mathscr{B}\setminus\{x^\gamma\}\big) \setminus\{x^\alpha\}$.
 
\begin{algorithm}[!ht]
\caption{The function to remove monomials from a Borel set.}
\label{alg:removeWithoutRepetions}
\begin{algorithmic}[1]
\STATE $\textsc{Remove}(\mathscr{B},q,x^\beta)$
\REQUIRE $\mathscr{B}$, a Borel set.
\REQUIRE $q$, a non-negative integer.
\REQUIRE $x^\beta$, a monomial.
\ENSURE the set of all Borel sets obtained from $\mathscr{B}$ removing in all the possible ways $q$ monomials. The monomial $x^\beta$ is to avoid repetitions: it will be 1 when the function is called by \textsc{BorelIdeals}, whereas it will be the last monomial removed whenever the function is called by itself.
\IF{$q = 0$}
\RETURN $\{ \mathscr{B} \}$;
\ELSE
\STATE $\textsf{borelSets} \leftarrow \emptyset$;
\STATE\label{line:minimals} $\textsf{minimalMonomials} \leftarrow \textsc{MinimalElements}(\mathscr{B})$;
\FORALL{$x^\alpha \in \textsf{minimalMonomials}$}
\IF{$x^\alpha >_{\texttt{DegLex}} x^\beta$}
 \STATE\label{line:remove} $\textsf{borelSets} \leftarrow{}$\parbox[t]{5cm}{$\textsf{borelSets} \cup {}$\\ $\textsc{Remove}\big(\mathscr{B}\setminus\{x^{\alpha}\},q-1,x^{\alpha}\big)$;}
 \ENDIF
\ENDFOR
\RETURN \textsf{borelSets};
\ENDIF
\end{algorithmic}
\end{algorithm}

\begin{algorithm}[!ht]
\caption{Auxiliary functions}
\label{alg:aux}
\begin{algorithmic}
\STATE $\textsc{GotzmannNumber}\big(p(t)\big)$
\REQUIRE $p(t)$, a Hilbert polynomial.
\ENSURE the Gotzmann number of $p(t)$. 
\end{algorithmic}

\medskip

\begin{algorithmic}
\STATE $\textsc{MinimalElements}(\mathscr{B})$
\REQUIRE $\mathscr{B}$, a Borel set.
\ENSURE the set of minimal elements of $\mathscr{B}$. 
\end{algorithmic}
\end{algorithm}

\subsection{The implementation}

The key of an efficient implementation is how we realize the Borel set. The guidelines are
\begin{itemize}
\item slim structure, in order to take up less memory as possible (since the number of final objects can be huge);
\item quick implementation of the basic operations:
\begin{enumerate}
\item embedding of a Borel set in a poset with one more variable (lines \ref{line:embed1}-\ref{line:embed2} \textbf{Algorithm} \ref{alg:coreDFS});
\item computation of the minimal elements of a Borel set (line \ref{line:minimals} \textbf{Algorithm} \ref{alg:removeWithoutRepetions});
\item removal of a monomial from a Borel set (line \ref{line:remove} \textbf{Algorithm} \ref{alg:removeWithoutRepetions});
\item computation of the saturation of the ideal generated by a Borel set (line \ref{line:saturation} \textbf{Algorithm} \ref{alg:coreDFS}).
\end{enumerate}
\end{itemize}

The most compact way to describe a Borel set is to store its minimal elements. This is efficient also to evaluate the belonging to the Borel set of any other monomial, because we can use the following proposition.
\begin{proposition}[{\cite[Proposition 2.3]{LellaDeformations}}]
Let $x^\alpha, x^\beta$ be two monomials in $\K[x_0,\ldots,x_n]_r$.
\begin{equation}
x^\alpha >_B x^\beta \ \Longleftrightarrow\ \sum_{j=i}^n (\alpha_j - \beta_j) \geqslant 0,\ \forall\ i=0,\ldots,n.
\end{equation}
\end{proposition}
Indeed if a monomial $x^\gamma$ belongs to the Borel set, there exists a minimal element $x^\alpha$ such that $x^\gamma \geq_B x^\alpha$, whereas if $x^\gamma$ does not belong there exists a minimal element $x^\beta$ such that $x^\gamma <_B x^\beta$.

In this way, the computation of the set $S$ of minimal elements is immediate. When a minimal monomial $x^\alpha$ is removed, we computed all the monomials that can be obtained from $x^\alpha$ with an increasing elementary move and for each of them we check if it belongs to the Borel set corresponding to $S\setminus\{x^\alpha\}$: if not we add such monomial to the set of minimal elements.

The description of the Borel set by means of its minimal elements turns out to be efficient also when we need to pass to a poset with one more variable. Indeed given the Borel set $\overline{\mathscr{B}} \subset \pos{n-k}{r}$ with minimal elements stored in the set $\overline{M}$, by Lemma \ref{lem:partition} we know that the Borel set $\mathscr{B} = \big\{\big(\langle\overline{\mathscr{B}}\rangle^\sat\cdot\K[x_{k-1},\ldots,x_n]\big)_r\big\} \subset \pos{n-k+1}{r}$ is described by the set of minimal elements
\begin{equation}
M = \left\{ x_n^{\alpha_n} \cdots x_{k+1}^{\alpha_{k+1}} x_{k-1}^{\alpha_k}\ \vert\ x_n^{\alpha_n} \cdots x_{k+1}^{\alpha_{k+1}} x_{k}^{\alpha_k} \in \overline{M}\right\}.
\end{equation}

A not trivial task is to compute the number of monomials that do not belong to the Borel set after the embedding in a poset with one more variable. In principle it consists in the computation of the Hilbert polynomial and in its evaluation at degree $r$. But the algebraic approach (used in \cite{MooreNagel}) slows down the algorithm because any software dedicated to the study of polynomial ideals computes the Hilbert polynomial with Gr\"obner basis tools, which in this case are unnecessary.

Given a Borel set $\overline{\mathscr{B}} \subset \pos{n-k}{r}$, let $\overline{\mathscr{N}} = \overline{\mathscr{B}}^{\mathcal{C}}$. If we decompose $\overline{\mathscr{N}}$ according to the power of the smallest variable $x_k$, i.e.
\begin{equation}
\overline{\mathscr{N}} = \coprod_{i=0}^r \overline{\mathscr{N}}_i,\qquad \overline{\mathscr{N}}_i = \left\{x^{\alpha} \in \overline{\mathscr{N}} \text{ s.t. } \alpha_k = i \right\}.
\end{equation}
We can compute the cardinality of $\mathscr{N} = \mathscr{B}^{\mathcal{C}}$ rewriting \eqref{eq:cardN} of Lemma \ref{lem:countingConditions} as
\[
\left\vert \mathscr{N} \right\vert = \sum_{\begin{subarray}{c} x^{\alpha} \in \overline{\mathscr{N}} \\ x^{\alpha}=x_n^{\alpha_n}\cdots x_{k}^{\alpha_k}\end{subarray}} (\alpha_k + 1) = \sum_{i=0}^r \sum_{\begin{subarray}{c} x^{\alpha} \in \overline{\mathscr{N}}_i \\ x^{\alpha}=x_n^{\alpha_n}\cdots x_{k}^{i}\end{subarray}} (i + 1) = \sum_{i=0}^r (i+1)\vert \overline{\mathscr{N}}_i\vert.
\]
Hence we add to the structure describing a Borel set an array of $r+1$ integers, such that the $i$-th index is equal to the number of monomials not belonging to the Borel set with the power of the smallest variable equal to $i$. It is easy also to deduce the array corresponding to a Borel set after the embedding in a poset with one more variable, indeed any $x^\alpha \in \overline{\mathscr{N}}_i$ implies the belonging to $\mathscr{N}$ of $i+1$ monomials
\[
\begin{split}
& x^\alpha = x_n^{\alpha_n}\cdots x_{k}^{\alpha_k} x_{k-1}^0 \in \mathscr{N}_0,\\
& \down{k}(x^\alpha) = x_n^{\alpha_n}\cdots x_{k}^{\alpha_k-1} x_{k-1} \in \mathscr{N}_1,\\ 
& \qquad\vdots \\
& \alpha_k \down{k}(x^\alpha) = x_n^{\alpha_n}\cdots x_{k}^{0} x_{k-1}^{\alpha_k} \in \mathscr{N}_{\alpha_k}
\end{split}
\]
so that
\[
\vert\mathscr{N}_j\vert = \sum_{i=0}^j \vert\overline{\mathscr{N}}_i\vert.
\]
Moreover when a minimal monomial $x^\alpha$ is removed from a Borel set, we increase by one the index of the array corresponding to the power of the smallest variable in $x^\alpha$.

The last operation which has to be as quick as possible is the computation of the saturation of the ideal. We exploit the algebraic approach of \cite{MooreNagel} and we add to the structure describing the Borel set the list of monomials that generate the saturated ideal corresponding to the Borel set. If $x^\alpha = x_n^{\alpha_n}\cdots x_{k}^{\alpha_k}$ is a minimal monomial of $\mathscr{B}$ then $x^{\underline{\alpha}} = x^\alpha\vert_{x_{k}=1}= x_n^{\alpha_n}\cdots x_{k+1}^{\alpha_{k+1}}$ is a generator of the saturated ideal. As shown in \cite{MooreNagel}, removing $x^\alpha$ from $\mathscr{B}$ implies that $\langle \mathscr{B}\setminus\{x^{\alpha}\}\rangle^\sat$ is generated by the same generators of $\langle\mathscr{B}\rangle^\sat$ except for $x^{\underline{\alpha}}$ that is replaced by
\[
x^{\underline{\alpha}}\cdot x_m,\ \ldots\ ,x^{\underline{\alpha}}\cdot x_1,\qquad\text{where } x_m = \min x^{\underline{\alpha}}.
\]

\subsection{Experimental results}
The following experimental results are obtained with an implementation of the algorithm realized with the guidelines showed above, coded in \texttt{java} and ran on a MacBook Pro with a 2.4 GHz Intel Core 2 Duo processor. It can be tested by means of an applet available at 
\begin{center}
\href{http://www.personalweb.unito.it/paolo.lella/HSC/borelGenerator.html}{\texttt{www.personalweb.unito.it/paolo.lella/HSC/borelGenerator.html}}.
\end{center}
In the following tables we consider Hilbert polynomials of degree 0,1,2 (i.e. points, curves and surfaces) and projective spaces of increasing dimension and we reported the elapsed time (in seconds) of the computation and the number of Borel-fixed ideals obtained. 

\begin{table}[H]
\begin{center}
\begin{tabular}{c | c | c | c | c}
Time (sec)\phantom{$\Big\vert$} & $n=5$ & $n=10$ & $n=15$ & $n=20$ \\
\hline
$p(t) = 5$\phantom{$\Big\vert$} & 0.101 & 0.025 & 0.017 & 0.021 \\
\hline
$p(t) = 10$\phantom{$\Big\vert$} & 0.062 &	0.064 &	0.119 &	0.048 \\
\hline
$p(t) = 15$\phantom{$\Big\vert$} & 0.079 &	0.225 &	0.298 &	0.401 \\	
\hline
$p(t) = 20$\phantom{$\Big\vert$} & 0.341 &	1.595 &	2.735 &	3.870 \\
\hline
$p(t) = 25$\phantom{$\Big\vert$} & 2.094 &	13.595 &	24.497 &	33.303 \\
\end{tabular}

\bigskip

\begin{tabular}{c | c | c | c | c}
Ideals\phantom{$\Big\vert$} & $n=5$ & $n=10$ & $n=15$ & $n=20$ \\
\hline
$p(t) = 5$\phantom{$\Big\vert$} & 5 & 5 & 5& 5 \\
\hline
$p(t) = 10$\phantom{$\Big\vert$} & 42& 50& 50& 50\\
\hline
$p(t) = 15$\phantom{$\Big\vert$} & 287& 417& 425&425 \\
\hline
$p(t) = 20$\phantom{$\Big\vert$} & 1732& 3130& 3263& 3271\\
\hline
$p(t) = 25$\phantom{$\Big\vert$} & 9501& 21616& 23158& 23291\\
\end{tabular}
\end{center}
\caption{\label{tab:deg0} Experiments with constant Hilbert polynomials. The Gotzmann number coincides with the number of points.}
\end{table}

\begin{table}[H]
\begin{center}
\begin{tabular}{c | c | c | c | c}
Time (sec)\phantom{$\Big\vert$} & $n=5$ & $n=10$ & $n=15$ & $n=20$ \\
\hline
$5t+1\ (11)$\phantom{$\Big\vert$} & 0.117 & 0.159 &	0.067 &	0.0621 \\	
\hline
$5t+7\ (17)$\phantom{$\Big\vert$} & 0.502 &	1.480 &	2.312 &	3.290 \\
\hline
$5t+13\ (23)$\phantom{$\Big\vert$} & 10.513 &	56.652 & 91.456 &	128.341 \\	
\hline
$8t-6\ (22)$\phantom{$\Big\vert$} & 0.987&	2.623	&4.138	&5.852 \\
\hline
$8t-3\ (25)$\phantom{$\Big\vert$} & 3.008 &	14.128 &	22.960 &	32.300 \\
\hline
$8t\ (28)$\phantom{$\Big\vert$} & 12.960 &	72.053 &	273.719& 	238.856\\
\end{tabular}

\bigskip

\begin{tabular}{c | c | c | c | c}
Ideals\phantom{$\Big\vert$} & $n=5$ & $n=10$ & $n=15$ & $n=20$ \\
\hline
$5t+1\ (11)$\phantom{$\Big\vert$} & 89 & 98 &	98 &	 98 \\	
\hline
$5t+7\ (17)$\phantom{$\Big\vert$} & 3028 &	 4560 &	4587 &	4587 \\
\hline
$5t+13\ (23)$\phantom{$\Big\vert$} & 58124 &	123689 & 126962 &	127030 \\
\hline
$8t-6\ (22)$\phantom{$\Big\vert$} & 4171 &	 6741	& 6837	& 6838 \\
\hline
$8t-3\ (25)$\phantom{$\Big\vert$} & 17334 &	32073 &	32848 &	32868 \\
\hline
$8t\ (28)$\phantom{$\Big\vert$} & 68291 &	144660 &	149777& 	149976\\
\end{tabular}
\end{center}
\caption{\label{tab:deg1} Experiments with Hilbert polynomials of curves. The Gotzmann number is reported in brackets.}
\end{table}

\begin{table}[H]
\begin{center}
\begin{tabular}{c | c | c | c | c}
Time (sec)\phantom{$\Big\vert$} & $n=5$ & $n=10$ & $n=15$ & $n=20$ \\
\hline
$2t^2+8t-46\ (16)$\phantom{$\Big\vert$} & 0.312 & 0.189 &	0.304 &	0.516\\	
\hline
$2t^2+8t-42\ (20)$\phantom{$\Big\vert$} & 0.103 &	0.338 &	0.558 &	0.883\\
\hline
$2t^2+8t-38\ (24)$\phantom{$\Big\vert$} & 0.741& 	3.167 &	5.237 &	7.055\\
\hline
$4t^2-12t+10\ (20)$\phantom{$\Big\vert$} & 0.147 &	0.280 &	0.377 & 0.561 \\
\hline
$4t^2-12t+14\ (24)$\phantom{$\Big\vert$} & 0.953 &	3.909 &	6.007 &	8.588\\
\hline
$4t^2-12t+18\ (28)$\phantom{$\Big\vert$} & 9.066 &	50.071 &	82.592 &	112.237\\

\end{tabular}

\bigskip

\begin{tabular}{c | c | c | c | c}
Ideals\phantom{$\Big\vert$} & $n=5$ & $n=10$ & $n=15$ & $n=20$ \\
\hline
$2t^2+8t-46\ (16)$\phantom{$\Big\vert$} & 834 & 38 &	38 &	 38 \\	
\hline
$2t^2+8t-42\ (20)$\phantom{$\Big\vert$} & 481 &	 670 &	671 &	671 \\
\hline
$2t^2+8t-38\ (24)$\phantom{$\Big\vert$} & 4774 &	8393 & 8476 & 8476 \\
\hline
$4t^2-12t+10\ (20)$\phantom{$\Big\vert$} & 631 &	 856	& 857	& 857 \\
\hline
$4t^2-12t+14\ (24)$\phantom{$\Big\vert$} & 6394 &	10986 &	11082 &	11082 \\
\hline
$4t^2-12t+18\ (28)$\phantom{$\Big\vert$} & 51527 &	112852 &	115295& 115332\\
\end{tabular}
\end{center}
\caption{\label{tab:deg2} Experiments with Hilbert polynomials of degree 2. The Gotzmann number is reported in brackets.}
\end{table}

\bibliographystyle{amsplain}

\providecommand{\bysame}{\leavevmode\hbox to3em{\hrulefill}\thinspace}
\providecommand{\MR}{\relax\ifhmode\unskip\space\fi MR }
\providecommand{\MRhref}[2]{%
  \href{http://www.ams.org/mathscinet-getitem?mr=#1}{#2}
}
\providecommand{\href}[2]{#2}

\end{document}